\documentclass[12pt,onecolumn,draftcls]{IEEEtran}
\usepackage{amsmath,epsfig}
\usepackage{amssymb}
\usepackage{amsmath}
\usepackage{cases} 
\usepackage{amssymb,amsmath,cite}
\usepackage{epsfig}
\usepackage{color}
\usepackage{slashbox}
\usepackage{bm}
\usepackage{slashbox}
\newcommand{\SINR}{{\mbox{SINR}}}

\newcommand{\bc}{\mathbf{c}}
\newcommand{\be}{\mathbf{e}}
\newcommand{\bE}{\mathbf{E}}

\newcommand{\bg}{\mathbf{g}}

\newcommand{\bI}{\mathbf{I}}

\newcommand{\bx}{\mathbf{x}}

\newcommand{\bs}{\mathbf{s}}
\newcommand{\bt}{\mathbf{t}}

\newcommand{\bA}{\mathbf{A}}

\newcommand{\bp}{\mathbf{p}}
\newcommand{\bq}{\mathbf{q}}

\renewcommand{\frac}{\dfrac}

\renewcommand{\Pr}{\mathbb{P}}

\newcommand{\K}{{\cal K}}
\newcommand{\N}{{\cal N}}

\newcommand{\I}{{\cal I}}
\newcommand{\J}{{\cal J}}

\renewcommand{\S}{\cal S}
\newcommand{\SI}{\mbox{SINR}}

\newcommand{\T}{{{T}}}

\newtheorem{dingli}{Theorem~}
\newtheorem{yinli}{Lemma~}
\newtheorem{mingti}{Proposition~}

\newenvironment{proof}[1][Proof]{\begin{trivlist}
\item[\hskip \labelsep {\bfseries #1}]}{\end{trivlist}}

\begin{document}
  \title{Sample Approximation-Based Deflation Approaches for Chance SINR Constrained Joint Power and Admission Control
  \author{Ya-Feng Liu, Mingyi Hong, {and Enbin Song}}
  \thanks{Part of this work has been presented in the 14th IEEE International Workshop on Signal Processing Advances in Wireless Communications (SPAWC), Darmstadt, Germany, June 16--19, 2013 \cite{SPAWC}.}
 \thanks{Y.-F.~Liu (corresponding author) is with the State Key Laboratory
of Scientific and Engineering Computing, Institute of Computational
Mathematics and Scientific/Engineering Computing, Academy of
Mathematics and Systems Science, Chinese Academy of Sciences,
Beijing, 100190, China (e-mail: yafliu@lsec.cc.ac.cn). Y.-F. Liu is supported by the National Natural
Science Foundation of China under Grants 11301516, 11331012, and 11571221.}
\thanks{M. Hong is with the Department of Industrial and Manufacturing Systems Engineering (IMSE), Iowa State University, Ames, IA 50011, USA (email: {mingyi@iastate.edu}). M. Hong is supported in part by NSF under Grant CCF-1526078 and by AFOSR under grant 15RT0767.}
\thanks{E. Song is with the Department of
Mathematics, Sichuan University,
Chengdu, 610064, China (e-mail:
 {e.b.song}@163.com). E. Song is supported by the National Natural
Science Foundation of China under Grant 61473197.}}

  \maketitle

\begin{abstract}
     \boldmath  
{}{Consider the joint power and admission control (JPAC) problem for a multi-user single-input single-output (SISO) interference channel.}
{Most existing works on JPAC assume the perfect instantaneous channel state information (CSI). In this paper, we consider the JPAC problem with the imperfect CSI, that is, we assume that only the channel distribution information (CDI) is available. We formulate the JPAC problem into a chance (probabilistic) constrained program, where each link's SINR outage probability is enforced to be less than or equal to a specified tolerance. To circumvent the computational difficulty of the chance SINR constraints, we propose to use the sample (scenario) approximation scheme to convert them into finitely many simple linear constraints. Furthermore, we reformulate the sample approximation of the chance SINR constrained JPAC problem as a composite group sparse minimization problem and then approximate it by a second-order cone program (SOCP). The solution of the SOCP approximation can be used to check the simultaneous supportability of all links in the network and to guide an iterative link removal procedure ({}{the deflation approach}). We exploit the special structure of the SOCP approximation and custom-design an efficient algorithm for solving it. Finally, we illustrate the effectiveness and efficiency of the proposed sample approximation-based deflation approaches by simulations.}

 \end{abstract}

\section{Introduction}
Joint power and admission control (JPAC) has been recognized as an effective tool for interference
management in cellular, ad hoc, and cognitive underlay wireless
networks for two decades. {Generally speaking, there are two kinds of JPAC: one is to support a maximum number of links at their specified signal to interference plus noise ratio (SINR) targets while using minimum total transmission power when all links in the network cannot be simultaneously supported\cite{convex_approximation,region,removals,Simple,ex2,ex3,bbound,SPAWC,msp,powerx,scheduling,decentralized,admission_luo,performance,r1,r2,r3,r4,lpadmission,distributed,q-norm,chang}, 
and the other is to determine whether a new arrival link can be admitted to the network while maintaining the SINR of all already admitted links above their required SINR levels \cite{admission_Bambos,ex1,ACM}. 
This paper focuses on the former one,} {}{which not only determines the set of links that must be turned off and rescheduled (possibly along orthogonal resource
dimensions such as time, space, or frequency slots)}, {but also {}{alleviates} the difficulties
of the convergence} of stand-alone power control algorithms. For example, a longstanding issue associated with {the Foschini-Miljanic} algorithm \cite{Simple} is that{,} it does not converge when the preselected SINR levels are {infeasible}. In this case, a JPAC approach must be adopted to determine which links {to} be removed.

The JPAC problem can be solved to global optimality by checking the
simultaneous supportability of every subset of links. However, the
computational complexity of this enumeration approach grows
exponentially with the total number of links. Theoretically, the problem is known to be NP-hard to solve (to global optimality) {and} to approximate (to constant ratio of global optimality) \cite{convex_approximation,removals,msp}, so various heuristic
algorithms have been proposed \cite{convex_approximation,region,removals,Simple,ex2,ex3,bbound,msp,SPAWC,powerx,scheduling,decentralized,admission_luo,performance,r1,r2,r3,r4,lpadmission,distributed,q-norm,chang}. 
In particular,
the reference~\cite{convex_approximation} proposed a convex approximation-based algorithm, {called linear
programming deflation (LPD) algorithm. Instead of solving the original NP-hard problem directly, the LPD algorithm
solves an appropriate LP approximation of the original problem at each iteration and use its solution to guide the removal of interfering links.} The removal procedure is terminated if all the remaining links in
the network are simultaneously supportable. The reference~\cite{msp} developed another LP approximation-based new linear programming deflation (NLPD) algorithm for the JPAC problem. In \cite{msp}, the JPAC problem is first equivalently reformulated as a sparse $\ell_0$-minimization problem and then its $\ell_1$-convex approximation is used to derive a LP, which is different from the one in \cite{convex_approximation}. Again, the solution to the derived LP can guide an iterative link removal procedure, and the removal procedure is terminated if all the remaining links in
the network are simultaneously supportable. {Similar ideas were also used in \cite{admission_luo,r2,decentralized} to solve the joint beamforming and admission control problem for the cellular downlink network.}

{{Most of the aforementioned works on the joint power/beamforming and admission control problem assume the perfect instantaneous channel state information (CSI) except \cite{convex_approximation,removals,powerx,r4}.} 
In \cite{convex_approximation}, the authors also considered the worst-case robust JPAC problem with bounded channel estimation errors. The key in \cite{convex_approximation} is that the LP approximation with bounded uncertainty can be equivalently rewritten as a second-order cone program (SOCP). The overall approximation algorithm remains similar to LPD for the case of the perfect CSI, except that the SOCP formulation is used to carry out power control and its solution is used to check whether links are simultaneously supportable in the worst case. In \cite{removals,powerx,r4}, 
the authors employed the Foschini-Miljanic algorithm \cite{Simple} or its variants to update the power and then use the updated power to guide the removals of links without assuming the perfect CSI (as long as the SINR can be measured at the receiver and feedbacked to the corresponding transmiter). The Foschini-Miljanic algorithm \cite{Simple} can leverage the perfect CSI assumption when updating the power, but it does not take admission control into consideration compared to the disciplined convex approximation-based power control algorithms in \cite{convex_approximation,msp,admission_luo,r2,decentralized}. This makes the JPAC algorithms where the power is updated by the Foschini-Miljanic algorithm suffer a significant performance loss in the number of supported links compared to those where the power is updated by the disciplined convex approximation-based power control algorithms.}


%

{The assumption of the perfect CSI generally does not hold true due to CSI estimation errors or limited CSI feedback in practice \cite{Outage12,feedback}. Even though the instantaneous CSI can be perfectly available, dynamic JPAC in accordance with its variations would lead to excessively high
computational and signaling costs.} 
{}{In this paper, we consider the chance (probabilistic or outage-based) SINR constrained JPAC problem, where each link's SINR outage probability must be kept below a given tolerance.} {{}{Different from most of the aforementioned works on JPAC where the perfect CSI is assumed, our new formulation only requires the availability of the channel distribution information (CDI).} Due to the fact that the CDI can remain unchanged over a relatively long period of time, JPAC based on the CDI can therefore be performed on a relatively slow timescale (compared to fast fluctuations of instantaneous channel conditions), hence the overall computational cost and signaling overhead can be significantly reduced, {which is particularly appealing from the network operator's perspective}. Moreover, the chance SINR constrained JPAC formulation can maximize the number of \emph{long-term} supported links by using minimum total transmission power, and at the same time guarantee that \emph{short-term} SINR requirements are respected with high probability, which depends on the user-specified outage tolerance.

It is well-known that characterizing Quality-of-Service (QoS) constraints in terms of an outage probability can significantly improve practicality of the resource allocation algorithms} {including power control and beamforming design; see \cite{distributionally,Outage6,Outage1,Outage12,long-term,Outage13,Outage14,outagenew1} 
and references therein.} Therefore, the chance constrained programming methodology has been widely applied to wireless system designs in recent years.
%
{However, as far as we know, such methodology has not been used in the context of JPAC. This is largely due to the computational challenge of solving the chance SINR constrained JPAC problem. First, chance SINR constraints do not have closed-form expressions and are nonconvex in general. Second, even when the CSI is perfectly available, the JPAC problem is NP-hard to solve and to approximate \cite{convex_approximation,removals,msp}.

{}{This is the first work that formulates the chance SINR constrained JPAC problem and proposes efficient deflation approaches for solving it.} The main contributions of this paper are twofold.

\begin{itemize}
 {}{
 \item \emph{Novel Problem Formulation and Reformulation.} In this paper, we assume that only the CDI is available, which is different from most of the existing works on JPAC where the perfect CSI is assumed. {}{We propose the first chance SINR constrained JPAC formulation}, where each link's SINR outage probability is required to be less than or equal to a preselected tolerance. Furthermore, we approximate the chance SINR constraint via sampling \cite{scenario,uncertainty} and reformulate the sample approximation of the chance SINR constrained JPAC problem {as} a composite group sparse minimization problem.



  \item \emph{Efficient Deflation Approaches.} We propose an efficient convex SOCP approximation (different from that in \cite{convex_approximation}) of the {}{group} sparse minimization reformulation. The solution of the SOCP approximation can be used to check the simultaneous supportability of all links in the network and to guide an iterative link removal procedure (the deflation approach). Instead of relying on standard SOCP solvers to solve the derived SOCP, we exploit its special structure and custom-design an efficient algorithm for solving it. {}{Note that the standard SOCP solvers cannot efficiently solve the SOCP approximation here because both the number of constraints and unknown variables in the SOCP approximation increase linearly with the number of samples, which is generally large in order to guarantee the approximation performance.}

}
\end{itemize}

\emph{Notation}. We {denote} the index
set $\{1,2,\ldots,K\}$ by ${\K}$. Lowercase boldface and uppercase
boldface are used for vectors and matrices, respectively. For a
given vector $\bx,$ the {}{notation} $\max\{\bx\}$, $\min\left\{\bx\right\},$ $(\bx)_k,$ and $\|\bx\|_0$ stand
for its maximum entry, its minimum entry, its $k$-th entry, and the indicator function of $\bx$ (i.e., $\|\bx\|_0=0$ if $\bx=\bm{0}$ and $\|\bx\|_0=1$ otherwise){,} respectively.
%
 The expression 
 $\max\left\{\bx_1,\bx_2\right\}$ ($\min\left\{\bx_1,\bx_2\right\}$) represents the component-wise maximum (minimum) of two vectors $\bx_1$ and $\bx_2$.
~For any subset $\I\subseteq {\cal K}$,  $\bA_\I$ stands for
the matrix formed by the rows of $\bA$ indexed by $\I$. 
We use $(\bA_1,\bA_2)$ to denote the matrix formed by stacking matrices $\bA_1$ and $\bA_2$ by column and use $(\bA_1;\bA_2)$ to denote the matrix formed by stacking $\bA_1$ and $\bA_2$ by row. Similar {}{notation applies} to stacking of vectors and scalars. {Finally, we use $\be$ to represent the vector with all components
being one, $\bI$ the identity matrix, and $\bE_k$ the matrix with all entries being zero except its $k$-th column entries being one, respectively.}

\section{Review of the NLPD algorithm}\label{secreview}

{The algorithms developed for the chance SINR constrained JPAC problem in this paper are based on
  the NLPD algorithm \cite{msp} for the JPAC problem that assumes the perfect CSI. To} streamline the presentation, we briefly review the NLPD algorithm in this section. The basic idea of the NLPD algorithm is to
update the power and check whether all links can be simultaneously supported or
not. If the answer is yes, then terminate the algorithm; else drop one link from the
network and update the power again. The above process is repeated
until all the remaining links can be simultaneously supported.

Specifically, consider a $K$-link (a link corresponds to a transmitter-receiver
pair) {}{single-input single-output} interference channel with channel gains $g_{k,j}\geq0$ (from
transmitter $j$ to receiver $k$), noise power $\eta_k>0,$ SINR
target $\gamma_k>0,$ and power budget $\bar p_k>0$ for $k, j\in
{\K}$. Denote the power allocation vector by
$\bp=(p_1,p_2,\ldots,p_K)^\T$ and the power budget vector by $\bar
\bp=(\bar p_1,\bar p_2,\ldots,\bar p_K)^\T$. Treating interference
as noise, we can write the SINR at the $k$-th receiver as
\begin{equation}\label{sinr}\displaystyle
\SI_k(\bp)=\frac{g_{k,k}p_k}{\eta_k+\displaystyle\sum_{j\neq
k}g_{k,j}p_j},\quad\forall~k\in\K.\end{equation}

Correspondingly, we introduce an equivalent normalized channel. In particular, we use \begin{equation}\label{qq}\bq=\left(q_1,q_2,\ldots,q_K\right)^T\end{equation}
with
$q_k={p_k}/{\bar p_k}$
~to denote the normalized power allocation vector, and use
${\bc}=\left(c_1,c_2,\ldots,c_K\right )^T$ with
  $c_k={\left(\gamma_k\eta_k\right)}/{\left(g_{k,k}\bar p_k\right)}>0$~
to denote the normalized noise vector. 
We denote the normalized channel matrix by $\bA\in {\mathbb R}^{K\times
K}$ with its $(k,j)$-th entry
\begin{equation*}\label{A}
a_{k,j}=\left\{\begin{array}{cl}
1,&\text{if~}k=j;\\
\displaystyle - \frac{\gamma_kg_{k,j}\bar p_j}{g_{k,k}\bar p_k},&\text{if~}k\neq j.
\end{array}
\right.
\end{equation*}
With these {}{notation}, it is simple to check that 
$\SI_k(\bp)\geq\gamma_k$ if and only if $\left(\bA\bq-\bc\right)_k \geq 0.$

Based on the Balancing Lemma\cite{performance}, we reformulate the JPAC problem {as} a sparse optimization problem
\begin{equation}\label{sparse2}
\begin{array}{cl}
\displaystyle \min_{\bq} & \displaystyle \sum_{k\in\K}\|(\bc-\bA\bq)_k\|_0+\alpha\bar\bp^T\bq \\
\mbox{s.t.} & \displaystyle 
\mathbf{0}\leq
\bq\leq \be.
\end{array}
\end{equation}
In the above, 
$\alpha$ is a parameter and $\be$ is the all-one vector of length $K$. For details on the choice of the parameter $\alpha,$ we refer the readers to \cite[Section III-B]{msp}.
Since problem \eqref{sparse2} is NP-hard \cite{convex_approximation}, we further consider
its $\ell_1$-convex approximation (which is equivalent to {an} LP; see \cite{msp})
\begin{equation}\label{ll1}
\begin{array}{cl}
\displaystyle \min_{\bq} & \displaystyle \sum_{k\in\K}\left|(\bc-\bA\bq)_k\right|+\alpha\bar\bp^T\bq\\
\mbox{s.t.} & \displaystyle 
\mathbf{0}\leq
\bq\leq \be.
\end{array}
\end{equation}
By solving \eqref{ll1}, we know whether
all links in the network can be simultaneously supported or not. 
If
not, we drop one link (mathematically, delete the corresponding row and column of
$\bA$ and the corresponding entry of $\bar\bp$ and $\bc$) from the network according to some removal strategy, and solve a reduced problem \eqref{ll1} until all the remaining links are supported.
\section{Problem Formulation}
{Consider the chance SINR constrained JPAC problem, where the channel gains $\left\{g_{k,j}\right\}$ 
in the SINR expression \eqref{sinr} are random variables. In this paper, we assume the distribution of $\left\{g_{k,j}\right\}$ {}{is} known. However, we do not assume any specific channel distribution, which is different from most of the existing works on outage probability constrained resource allocation for wireless systems \cite{Outage6,Outage1,Outage13}.
We also assume that all coordinations and computations are carried out by a central controller who knows the CDI of all links.}
Since $\left\{g_{k,j}\right\}$ 
in \eqref{sinr} are random variables, we need to redefine {the concept of a supported link}. We call link $k$ is supported if its outage probability is below a specified tolerance $\epsilon\in(0,1),$ i.e.,
\begin{equation}\label{chance}\mathbb{P}\left(\SINR_k(\bp)\geq \gamma_k\right)\geq 1-\epsilon,\end{equation}
where the probability is taken with respect to the random variables $\left\{g_{k,j}\right\}.$ 

{}{The chance SINR constrained JPAC problem aims to maximize the number of supported links while using minimum total transmission power. Mathematically, the problem can be formulated as
\begin{equation}\label{cMSP}
\begin{array}{cl}
\displaystyle \max_{{\bp,\,{\cal S}}} & \displaystyle |{\cal S}|-\alpha \be^T\bp\\
[5pt] \mbox{s.t.} & \displaystyle
\Pr\left(\SINR_k(\bp)\geq \gamma_k\right)\geq 1-\epsilon,~k\in\cal S\subseteq\cal K,\\
&\bm{0} \leq \bp\leq \bar\bp.
\end{array}
\end{equation} In the above, $\cal S$ denotes the set of supported links and $|{\cal S}|$ denotes its cardinality, i.e., the number of supported links; the parameter $\alpha$ balances the relative importance of the two goals, i.e., maximizing the number of supported links (the first term $|{\cal S}|$ in the objective) and minimizing the total transmission power (the second term $\be^T\bp$ in the objective).

To gain further understanding of formulation \eqref{cMSP}, we compare it with the following two-stage formulation. Specifically, the first stage maximizes the number of admitted links:
\begin{equation}\label{cMSP1}
\begin{array}{cl}
\displaystyle \max_{{\bp,\,{\cal S}}} & \displaystyle |\cal S| \\
[5pt] \mbox{s.t.} & \displaystyle
\Pr\left(\SINR_k(\bp)\geq \gamma_k\right)\geq 1-\epsilon,~k\in\cal S\subseteq\cal K,\\
&\bm{0} \leq \bp\leq \bar\bp.
\end{array}
\end{equation} We use ${{\cal S}_0}$ to denote the optimal solution for problem \eqref{cMSP1} and call it \emph{the maximum admissible set}.
Notice that the solution for \eqref{cMSP1}
might not be unique. The second stage minimizes the total
transmission power required to support the admitted links: 
\begin{equation}\label{cMSP2}
\begin{array}{cl}
 \displaystyle\min_{\bp} & \be^T\bp \\
  \mbox{s.t.} &  \Pr\left(\SINR_k(\bp)\geq \gamma_k\right)\geq 1-\epsilon,~k\in{{\cal S}_0},\\
&\displaystyle \bm{0} \leq \bp\leq \bar\bp. 
\end{array}
\end{equation}
Due to the choice of ${\cal S}_0,$ power control problem
\eqref{cMSP2} is feasible.

Although the above two-stage formulation (i.e., \eqref{cMSP1} and \eqref{cMSP2}) is intuitive and easy to understand, the formulation \eqref{cMSP} is better in terms of modeling the JPAC problem; see the following Theorem \ref{thm1}.
%
Theorem \ref{thm1} can be shown by a similar argument as used in \cite[Theorem 1]{msp}} and a detailed proof is provided in Section I of \cite{technicalreport}. 
}


\begin{dingli}\label{thm1}
%
{}{Suppose the parameter $\alpha$ satisfies
\begin{equation}\label{alpha1}0< \alpha<\alpha_1:=1/{\be^T\bar\bp}.\end{equation} Then the optimal value of problem \eqref{cMSP1} is $M$ if and only if the optimal value of problem \eqref{cMSP} lies in $(M-1,M).$ Moreover, suppose $({\S}^*, \bp^*)$ is the solution of problem \eqref{cMSP}. Then, ${\S}^*$ is a maximum admissible set and $\be^T\bp^*$ is the minimum total transmission power to support any maximum admissible set.
%
}
%
\end{dingli}

{}{Theorem \ref{thm1} states that the single-stage formulation (6) with $\alpha\in(0,\alpha_1)$ is equivalent to the two-stage formulation (7) and (8) in terms of finding the maximum admissible set. Moreover, it is capable of picking the maximum admissible set with minimum total transmission power from potentially multiple maximum admissible sets.}

In the rest of this paper, we develop sample approximation-based deflation approaches for (approximately) solving the chance SINR constrained JPAC problem \eqref{cMSP}.

\section{Sample Approximation and Reformulation}

In general, the chance SINR constrained optimization problem \eqref{cMSP} is difficult to solve exactly, since it is difficult to obtain the closed-form expression of \eqref{chance}. In this section, we first approximate the computationally intractable chance SINR constraint via sampling, and then reformulate the sample approximation of problem \eqref{cMSP} {as} a composite group sparse optimization problem. {}{Three distinctive advantages of the sample approximation scheme in the context of approximating the chance SINR constraint \eqref{chance} are as follows.}
{First, it works for general channel distribution models and thus is distributionally robust.}
Second, the sample approximation technique significantly simplifies problem \eqref{cMSP} by replacing the difficult chance SINR constraint with finitely many simple linear constraints {}{(depending on the sample size)}. Last but not the least, solving the sample approximation problem returns a solution to the original chance constrained problem with guaranteed performance\cite{scenario,uncertainty}.

{{}{It is worthwhile remarking that safe tractable approximation \cite{upperbound1,anthony} is an alternative approach to the sample approximation approach to dealing with the chance constraint. The safe tractable approximation approach builds an analytic upper bound of the probability for the chance constraint to be violated. The advantage of this line of approach over the sample approximation approach is that solving the deterministic analytic upper bound will return a feasible solution to the chance constraint for sure. However, to build such an analytic upper bound, some strict conditions on structures of the function composed in the chance constraint and on the distribution of the random variables are required.}}


\subsection{Sample Approximation} We handle the chance SINR constraint via sample approximations\cite{scenario,stochastic}.
{Suppose $\left\{g_{k,j}^n\right\}_{n=1}^N$ are $N$ independent samples drawn according to the distribution of $\left\{g_{k,j}\right\}$ by the central controller}, we use \begin{equation}\label{sample}
  {\SI}_k^n(\bp):=\frac{g_{k,k}^np_k}{\eta_k+\displaystyle\sum_{j\neq
k}g_{k,j}^np_j}\geq \gamma_k,~n\in\N:=\left\{1,2,\ldots,N\right\}
\end{equation} to approximate the chance SINR constraint \eqref{chance}. {}{Since the samples are random variables, the power allocation vector $\bp$  satisfying the sampled SINR constraints \eqref{sample} is also a random variable.} Intuitively, if the sample size $N$ is sufficiently large, then the power allocation vector $\bp$ satisfying \eqref{sample} will satisfy the chance SINR constraint \eqref{chance} with high probability.

%
%
%

{}{The above intuition has been rigorously shown in \cite[Theorem 1]{uncertainty} and \cite[Theorem 1]{distributionally}. It is shown that, if the sample size $N$ satisfies
\begin{equation}\label{N-K*}N\geq N^*:=\left\lceil\frac{1}{\epsilon}\left(K-1+\ln\frac{1}{\delta}+\sqrt{2(K-1)\ln\frac{1}{\delta}+\ln^2\frac{1}{\delta}}\right)\right\rceil\end{equation} for any $\delta\in(0,1),$ then any solution to the linear system \begin{equation}\label{linearsystem}\SINR_k^n(\bp)\geq \gamma_k,~k\in\K,~n\in\N\end{equation}
 will satisfy the chance SINR constraint \eqref{chance} for all $k\in\K$ with probability at least $1-\delta.$
In particular, if $\delta$ is chosen to be a very small value, any solution to \eqref{linearsystem} will almost surely be feasible for the chance SINR constraints \eqref{chance} for all $k\in\K.$ Note that the number of samples needed will not increase significantly as $\delta$ decreases, since $N^*$ has only a logarithmic dependence on $1/\delta.$ Although the dependence of $N^*$ on $\epsilon$ is $N^*=O(1/\epsilon),$ really small values of $\epsilon$ are of no interest in the scenario considered in this paper. 
}


{}{The linear system \eqref{linearsystem} might have multiple solutions. The most interesting solution is the one that minimizes the total transmission power, i.e., the solution to the following problem
%
%
\begin{equation}\label{power}
  \begin{array}{cl}
  \displaystyle\min_{\bp} & \be^T\bp\\
    \text{s.t.} &\SINR_k^n(\bp)\geq \gamma_k,~k\in\K,~n\in\N,\\
    &\bm{0}\leq \bp\leq\bar\bp.
  \end{array}
\end{equation}
Suppose $\bp$ is the solution to problem \eqref{power}. Then, for each $k\in{\K},$ there must exist an index $n_k\in \N$ such that \begin{equation}\label{balance}\SINR_k^{n_k}\left(\bp\right)=\gamma_k.\end{equation}}{For simplicity, we will refer link $k$ to be supported if all constraints in \eqref{sample} are satisfied in the sequel.}

\subsection{Sampled Channel Normalization}\label{sec-normalization}
To facilitate the reformulation of the sample approximation of problem \eqref{cMSP} and the development of {efficient algorithms}, we normalize the sampled channel parameters. {To this end, we use} 
$$\bc_k=\left(\frac{\gamma_k\eta_k}{g_{k,k}^1\bar p_k}, \frac{\gamma_k\eta_k}{g_{k,k}^2\bar p_k},\ldots,\frac{\gamma_k\eta_k}{g_{k,k}^N\bar p_k}\right)^T\in\mathbb{R}^{N\times 1}$$ to denote the normalized noise vector of link $k$. Define \begin{equation*}\label{A}
a_{k,j}^n=\left\{\begin{array}{cl}
1,&\text{if~}k=j;\\[3pt]
\displaystyle - \frac{\gamma_kg_{k,j}^n\bar p_j}{g_{k,k}^n\bar p_k},&\text{if~}k\neq j,
\end{array}
\right.
\end{equation*} $$\bm{a}_k^n=\left(a_{k,1}^n,a_{k,2}^n,\ldots,a_{k,K}^n\right)\in\mathbb{R}^{1\times K},~n\in\N,~k\in\K,$$and
$$\bA_k=\left(\bm{a}_k^1;\bm{a}_k^2;\ldots;\bm{a}_k^N\right)\in\mathbb{R}^{N\times K},~k\in\K.$$ Notice that the entries of the $k$-th column of $\bA_k$ are one, and all the other entries are nonpositive. This special structure of $\bA_k~(k\in\K)$ will play an important role in the following algorithm design.
%
Furthermore, we let $$\bc=\left(\bc_1;\bc_2;\ldots;\bc_K\right)\in\mathbb{R}^{NK\times1}~\text{and}~\bA=\left(\bA_1;\bA_2;\ldots;\bA_K\right)\in\mathbb{R}^{NK\times K}.$$

With the above {}{notation} and \eqref{qq}, we can see that 
$\SINR_k^n(\bp)\geq\gamma_k$ for all~$n\in\N$ if and only if $\bA_k\bq\geq\bc_k.$
Consequently, {the sample approximation of problem \eqref{cMSP}} can be equivalently rewritten as
\begin{equation}\label{subproblem}
\begin{array}{cl}
\displaystyle \max_{{\bq,\,{\cal S}}} & \displaystyle |{\cal S}| - \alpha\bar \bp^T\bq\\
[5pt] \mbox{s.t.} &\displaystyle \bA_k\bq-\bc_k\geq 0,\ k\in\cal S\subseteq\cal K,\\[3pt]
&\mathbf{0}\leq \bq\leq \be.
\end{array}
\end{equation}

\subsection{Composite Group Sparse Minimization Reformulation}\label{subsec:groupsparse}
By the definition of $\|\cdot\|_0,$ the sampled JPAC problem \eqref{subproblem} can be reformulated {as} the following composite group sparse optimization problem
\begin{equation}\label{csparse2}
\begin{array}{cl}
\displaystyle \min_{\bq} & \displaystyle \sum_{k\in\K}\|\max\left\{\bc_k-\bA_k\bq,\bm{0}\right\}\|_{0}+\alpha\bar\bp^T\bq \\
\mbox{s.t.} & \displaystyle \mathbf{0}\leq
\bq\leq \be.
\end{array}
\end{equation}
Problem \eqref{csparse2} has the following property stated in Proposition \ref{mingti1}, which is mainly due to the special structure of $\bA_k.$ The proof of Proposition \ref{mingti1} can be found in Appendix \ref{app-proposition}. 
\begin{mingti}\label{mingti1}Suppose that $\bq^*$ is the solution to problem \eqref{csparse2} and link $k$ is supported at the point $\bq^*$ (i.e., $\bA_k\bq^*\geq\bc_k$). Then there must exist an index $n_k\in \N$ such that $\left(\bc_k-\bA_k\bq^*\right)_{n_k}=0.$
\end{mingti}

 Proposition \ref{mingti1} implies that problem \eqref{csparse2} can be viewed as an (nontrivial) extension of problem \eqref{sparse2}. In fact, we know from Proposition \ref{mingti1} that when {$N=1$}, the solution of problem \eqref{csparse2} satisfies $\bc-\bA\bq^*\geq\bm{0}$ (i.e., $(\bc-\bA\bq^*)_k=0$ for supported links and $(\bc-\bA\bq^*)_k>0$ for unsupported links), and problem \eqref{csparse2} reduces to problem \eqref{sparse2}. Since problem \eqref{sparse2} is NP-hard to solve to global optimality and NP-hard to approximate to constant factor of global optimality\cite{convex_approximation,removals,msp}, it follows that problem \eqref{csparse2} is also NP-hard to solve and approximate.

{}{A key difference between problems \eqref{csparse2} and \eqref{sparse2} lies in the $\max$ operator introduced in problem \eqref{csparse2}. In problem \eqref{sparse2}, if link $k$ is supported, then $\bc_k-\bA_k\bq^*$ is a scalar and equals zero; while in problem \eqref{csparse2}, if link $k$ is supported, then $\bc_k-\bA_k\bq^*\leq \bm{0}$ but not necessarily equal to zero. Therefore, to correctly formulate the JPAC problem, we introduce a $\max$ operation and put $\max\left\{\bc_k-\bA_k\bq,\bm{0}\right\}$ in $\|\cdot\|_0$ instead of $\bc_k-\bA_k\bq;$ see problem \eqref{csparse2}. Notice that in the sparse formulation it is desirable that a link is supported if and only if the corresponding $\ell_0$-quasi-norm is zero. To further illustrate this, we give the following example, where $K=N=2,$ and
$$\bA=(\bA_1;\bA_2)=\left(
        \begin{array}{cc}
          1 & -0.2 \\
          1 & -0.5 \\
          -0.3 & 1 \\
          -0.5 & 1 \\
        \end{array}
      \right),
~\bc=\left(
       \begin{array}{c}
         0.5 \\
         0.5 \\
         0.5 \\
         0.5 \\
       \end{array}
     \right).$$
It can be checked that the only possible way to simultaneously support the two links $\left\{1,2\right\}$ is $\bq^*=\be$ and
$\max\left\{\bc_1-\bA_1\bq^*, \bm{0}\right\}=\max\left\{\bc_2-\bA_2\bq^*, \bm{0}\right\}=\bm{0}$ but $\bc_1-\bA_1\bq^*\neq \bm{0}$ and $\bc_2-\bA_2\bq^*\neq \bm{0}.$}

\section{Efficient Deflation Approaches for the Sampled JPAC Problem}
In this section, we develop efficient convex approximation-based deflation algorithms {}{for solving} the sampled JPAC problem \eqref{csparse2}. As can be seen, problem \eqref{csparse2} has a discontinuous objective function due to the first term. However, it allows for an efficient convex approximation. {We} first approximate problem \eqref{csparse2} by a convex problem, which is actually equivalent to {an} SOCP, and then design efficient algorithms for solving the approximation problem. The solution to the approximation problem can be used to check the simultaneous supportability of all links in the network and to guide an iterative link removal procedure (the deflation {}{approach}). We conclude this section with two convex approximation-based deflation algorithms for solving the sampled joint control problem \eqref{csparse2}.
\subsection{Convex Approximation}
{}{Recall that problem \eqref{csparse2} aims to find a feasible $\bq$ such that the vector $\bx=\left(\bx_1;\bx_2;\ldots;\bx_K\right)$ is as sparse as possible in the group sense, where $\bx_k:=\max\left\{\bc_k-\bA_k\bq,\bm{0}\right\}$. 
The nonsmooth mixed $\ell_2/\ell_1$ norm, $\sum_{k\in\K}\left\|\bx_k\right\|_{2},$ is shown in \cite{group} to be quite effective in characterizing and inducing the \emph{group sparsity}  of the vector $\bx.$
To understand this, observe that $\sum_{k\in\K}\left\|\bx_k\right\|_{2},$ the $\ell_1$ norm of the vector $\left(\|\bx_1\|_2,\|\bx_2\|_2,\ldots,\|\bx_K\|_2\right)^T,$ is a good approximation of its $\ell_0$ norm, which is equal to the $\ell_0$ norm of the vector $\left(\|\bx_1\|_0,\|\bx_2\|_0,\ldots,\|\bx_K\|_0\right)^T.$ More discussions on using the mixed $\ell_2/\ell_1$ norm to recover the group sparsity can be found in \cite{group}.}

{}{Motivated by the above discussion and the NP-hardness of problem \eqref{csparse2}, we consider the convex approximation of problem \eqref{csparse2} as follows:}
\begin{equation}\label{relax}
\begin{array}{cl}
\displaystyle \min_{\bq} & \displaystyle f(\bq):=\sum_{k\in\K}\left\|\max\left\{\bc_k-\bA_k\bq,\bm{0}\right\}\right\|_{2}+\alpha\bar\bp^T\bq \\
\mbox{s.t.} & \displaystyle \mathbf{0}\leq
\bq\leq \be.
\end{array}
\end{equation} {}{The convexity of the objective function of problem \eqref{relax} follows directly from \cite[Section 3.2.5]{boyd}.}

%

Compared to problem \eqref{csparse2}, the objective function of problem \eqref{relax} is continuous in $\bq$, but still nonsmooth. We give the subdifferential \cite{analysis} of the function $\left\|\max\left\{\bc_k-\bA_k\bq,\bm{0}\right\}\right\|_{2}$ in Proposition \ref{subgradient2}, which is important in the following analysis and algorithm design. 
The proof of Proposition \ref{subgradient2} is provided in Appendix \ref{app-1}.

\begin{mingti}\label{subgradient2}
  {}{Define $h_k(\bq)=\|\max\left\{\bc_k-\bA_k\bq,
  \bm{0}\right\}\|_2$. Suppose $\bc_k-\bA_k\bar\bq\leq\bm{0}$ and ${\N_k}^{=}:=\left\{n\,|\,(\bc_k-\bA_k\bar\bq)_n=0\right\}\neq\emptyset,$
  then $$\partial h_k(\bm{\bar\bq})=\left\{-\sum_{n\in\N_k^{=}}s_n\left(\bm{a}_k^n\right)^T\,|\,s_n\geq{0}, \sum_{n\in\N_k^{=}}s_n^2\leq 1\right\}.$$
  In particular, if $\N_k^{=}=\N,$ then $\partial h_k(\bm{\bar\bq})=\left\{-\bA_k^T\bs\,|\,\bs\geq\bm{0}, \|\bs\|_2\leq 1\right\}.$} {Further, if 
  ${\N_k}^{+}:=\left\{n\,|\,(\bc_k-\bA_k\bar\bq)_n>0\right\}\neq\emptyset$,} {then}
  \begin{equation}\label{differentiable}\begin{array}{rl}\nabla h_k(\bm{\bar\bq})\!\!&=\frac{-\sum_{n\in\N_k^+}\left(\bc_k-\bA_k\bar\bq\right)_n\left(\bm{a}_k^n\right)^T}{\|\max\left\{\bc_k-\bA_k\bar\bq, \bm{0}\right\}\|_2}=\frac{-\bA_k^T\max\left\{\bc_k-\bA_k\bar\bq, \bm{0}\right\}}{\|\max\left\{\bc_k-\bA_k\bar\bq, \bm{0}\right\}\|_2}.\end{array}\end{equation} 
\end{mingti}

We now discuss the choice of the parameter $\alpha$ in \eqref{relax}. The parameter $\alpha$ in \eqref{relax} should be chosen appropriately such that the following ``Never-Over-Removal'' property is satisfied: the solution of problem \eqref{relax} should simultaneously support all links at their desired SINR targets with minimum total transmission power as long as all links in the network are simultaneously supportable. Otherwise, since the solution of \eqref{relax} will be used to check the simultaneous supportability of all links and to guide the links' removal, it may mislead us to remove the links unnecessarily. {}{Notice that problem \eqref{relax} with $\alpha=0$ indeed can simultaneously support all links as long as the links are simultaneously supportable but not necessarily with minimum total transmission power, i.e., the solution of problem \eqref{relax} with $\alpha=0$ might not solve \eqref{power}.} {Theorem \ref{neverover}} gives {an interval} of the parameter $\alpha$ {to guarantee the ``Never-Over-Removal'' property.} {The proof of Theorem \ref{neverover} {(see Appendix \ref{app-neverover})} is mainly based on Proposition \ref{subgradient2}.}

\begin{dingli}\label{neverover}
  Suppose there exists some vector $\bq$ such that
$\bm{0}\leq\bq\leq\be$ and
  $\bA\bq\geq\bc${. Then} any solution of problem \eqref{relax} with \begin{equation}\label{alpha2}{}{0<\alpha\leq\alpha_2:=\frac{\min\left\{\bc\right\}}{K\max\left\{\bar\bp\right\}}}\end{equation} can simultaneously support all links at their desired SINR targets with minimum total transmission power.
\end{dingli}

Combining \eqref{alpha1} and \eqref{alpha2}, we propose to choose
the parameter $\alpha$ in \eqref{relax} according to
\begin{equation}\label{optalpha}\alpha=\min\left\{c_1\alpha_1,\,c_2\alpha_2\right\},\end{equation}
where $c_1, c_2\in (0,\, 1)$ are two constants.

{}{\emph{Link Removal Strategy.} The solution of problem \eqref{relax} can be used to guide the link removal process.
In particular, by solving \eqref{relax} with $\alpha$ given in
\eqref{optalpha}, we know whether all links in the network can be
simultaneously supported by simply checking if its solution $\bar\bq$ satisfies
$\bA\bar\bq\geq\bc$. Furthermore, if all links in the network cannot be simultaneously supported, we need to remove at least one link from the network. In particular, picking the worst sampled channel index $\bar n_k=\arg\max\left\{\bc_k-\bA_k\bar\bq\right\}$, we remove the link with the largest interference plus noise footprint
\begin{equation}\label{removalpower}k=\arg\max\left\{ \sum_{j\neq k}|a_{k,j}^{\bar n_k}|\bar q_j+\sum_{j\neq k}|a_{j,k}^{\bar n_j}|\bar q_k+c_k^{\bar n_k}\right\}.\end{equation}

In the next subsection, we design efficient algorithms to solve the convex but nonsmooth problem \eqref{relax}.}




\subsection{{Solution for Approximation Problem \eqref{relax}}}
By introducing {auxiliary variables} $\bx=\left(\bx_1;\bx_2;\ldots;\bx_K\right)$ and $\bt=(t_1;t_2;\ldots;t_K),$ 
problem \eqref{relax} can be transformed into the following SOCP
\begin{equation}\label{relax2}
\begin{array}{cl}
\displaystyle \min_{\bq,\bx,\bt} & \displaystyle \sum_{k\in\K}t_k+\alpha\bar\bp^T\bq \\
\mbox{s.t.} & \|\bx_k\|\leq t_k,~k\in\K{,}\\
&\bc-\bA\bq\leq\bx,\\
&\bm{0}\leq \bx,\\
&\displaystyle \mathbf{0}\leq
\bq\leq \be,
\end{array}
\end{equation} which can be solved by using the standard solver {like} CVX \cite{cvx}. {However, it is not an efficient way of solving problem \eqref{relax} by solving its equivalent SOCP reformulation \eqref{relax2}, since both the number of constraints and the number of unknown variables of problem \eqref{relax2} are of order $O(NK)$ while \eqref{N-K*} suggests that the sample size $N$ is generally very large.}


Next, we develop a custom-design algorithm for problem \eqref{relax} by first {smoothing the problem and then applying} the efficient projected alternate Barzilai-Borwein (PABB) algorithm \cite{BB_quad,BB2} {to solve its smooth counterpart}. {More specifically, we smooth problem \eqref{relax} by}
\begin{equation}\label{relax-appro}
\begin{array}{cl}
\displaystyle\min_{\bq} & \tilde f(\bq,\mu)=\displaystyle\sum_{k\in\K}\sqrt{\left\|\max\left\{\bc_k-\bA_k\bq,\bm{0}\right\}\right\|_{2}^2+\mu^2}+\alpha\bar\bp^T\bq \\
\mbox{s.t.} & \displaystyle \mathbf{0}\leq
\bq\leq \be,
\end{array}
\end{equation}where $\mu>0$ is the {smoothing} parameter. By \eqref{differentiable} in Proposition \ref{subgradient2}, the objective function $\tilde f(\bq,\mu)$ \ of problem \eqref{relax-appro} is differentiable everywhere with respect to $\bq$ and its gradient is given by
$$\nabla \tilde f(\bq,\mu)=\sum_{k\in\K}\frac{-\bA_k^T\max\left\{\bc_k-\bA_k\bq,\bm{0}\right\}}{\sqrt{\left\|\max\left\{\bc_k-\bA_k\bq,\bm{0}\right\}\right\|_{2}^2+\mu^2}}{+\alpha\bar\bp}.$$

{It can be shown that, as the parameter $\mu$ tends to zero, $\tilde f(\bq,\mu)$ \emph{uniformly} {}{converges to} $f(\bq)$ in \eqref{relax} and the solution of the smoothing problem \eqref{relax-appro} also converges to the one of problem \eqref{relax}; see Section II of \cite{technicalreport}. Therefore, when the parameter $\mu$ is very close to zero, the solution of problem \eqref{relax-appro} will be very close to the one of problem \eqref{relax}. 
%

{}{We apply the PABB algorithm \cite{BB_quad, BB2} to solve the {smoothing}
problem \eqref{relax-appro}. Three distinctive advantages of the PABB algorithm in the context of solving problem \eqref{relax-appro} are as follows. First, the box constraint is easy to project onto, and thus the PABB algorithm can be easily implemented to solve problem \eqref{relax-appro}. Second, the PABB algorithm requires only the gradient information but not the high-order derivative information, which makes it suitable for solving large-scale optimization problem \eqref{relax-appro}. Last but not least, the PABB algorithm enjoys a quite good numerical performance due to the use of the BB stepsize \cite{BB2}.}
When using the PABB algorithm to solve problem \eqref{relax-appro}, we
employ the \emph{continuation} technique \cite{continuation,program}. {}{That is, to obtain an
approximate solution of \eqref{relax}, we solve \eqref{relax-appro} with a series of gradually decreasing values for $\mu$, instead of using a tiny fixed $\mu$.} The continuation technique can reasonably improve the computational efficiency. {}{Solving problem \eqref{relax} by the PABB algorithm (combined with the smoothing and continuation techniques) is much faster than solving its SOCP reformulation \eqref{relax2} by the standard SOCP solver. Simulation results will be given later in Section \ref{sec-simulation}.}

\subsection{Convex Approximation-Based Deflation Algorithms}
The basic idea of the proposed convex approximation-based deflation algorithm for the sampled JPAC problem \eqref{csparse2} is to solve the power control problem
\eqref{relax} and check whether all links can be supported or not; if not, remove
a link from the network, and solve a reduced problem \eqref{relax} again until all the remaining
links are supported. 

{}{As in \cite{msp}, to accelerate the deflation procedure (avoid solving too many optimization problems in
the form of \eqref{relax}), we can derive an easy-to-check necessary
condition for all links in the network to be simultaneously
supported. It is easy to verify that the condition
\begin{equation}\label{necessary}
\bm{\mu}_+^\T\be-\left(\bm{\mu}_-^T\bc^{\max}+\be^\T\bc\right)\geq0
\end{equation}
is necessary for all links to be simultaneously
supported, where $\bm{\mu}_+=\max\left\{\bm{\mu},\mathbf{0}\right\},$ $\bm{\mu}_-=\max\left\{-\bm{\mu},\mathbf{0}\right\}$, $\bm{\mu}=\bA^T\be,$ and $\bc^{\max}=\left(\max\left\{\bc_1\right\};\max\left\{\bc_2\right\};\ldots;\max\left\{\bc_K\right\}\right).$
If
\eqref{necessary} is violated, we remove the link $k_0$ according to
\begin{equation}\label{smart}
k_0=\arg\max_{k\in {\cal K}}\left\{\sum_{j\ne k
}|\bar a_{k,j}|+\sum_{j\ne k}|\bar a_{j,k}|+\bar c_k\right\},
\end{equation}
which corresponds to applying the SMART rule \cite{removals} to the normalized sampled channel and substituting $\bq=\be.$ In \eqref{smart}, $\bar a_{k,j}$ and $\bar c_k$ are the averaged sample channel gain and noise, i.e.,
$$\bar {\bm{a}}_k=\left(\bar a_{k,1}, \bar a_{k,2},\ldots,\bar a_{k,K} \right)=\frac{\be^T\bA_k}{N},~\bar c_k=\frac{\be^T\bc_k}{N},~k\in\K.$$}

%

{The proposed convex approximation-based deflation {}{algorithmic} framework for problem \eqref{csparse2} is described in {Algorithm 1}.
{}{It is worthwhile remarking the {}{difference} between the proposed Algorithm 1 and the NLPD algorithm in \cite{msp}. The first key difference is that Algorithm 1 is designed for solving the sample approximation of {}{the} chance SINR constrained JPAC problem \eqref{cMSP} (i.e., problem \eqref{csparse2}) while the NLPD algorithm is designed for solving {}{the} instantaneous SINR constrained JPAC problem \eqref{sparse2}. As discussed in Subsection \ref{subsec:groupsparse}, problem \eqref{csparse2} includes problem \eqref{sparse2} as a special case. The second key difference between the two algorithms lies in the power control step (i.e., \textbf{Step 3}). More specifically, at each iteration, the proposed Algorithm 1 solves the SOCP \eqref{relax} to update the power while the NLPD algorithm solves the LP \eqref{ll1} to update the power. We also remark that the SOCP approximation \eqref{relax} used in Algorithm 1 is different from the one used in \cite{convex_approximation}. The two SOCP approximations take different forms and are derived from different perspectives.} 

\begin{center}
\framebox{
\begin{minipage}{14.5cm}
\flushright
\begin{minipage}{14.5cm}
\centerline{\bf Algorithm 1: A Convex Approximation-Based Deflation {}{Algorithmic} Framework}
\vspace{0.05cm} \textbf{Step 1.} Initialization: Input
data
$\left(\bA,\bc,\bar\bp\right).$\\[2.5pt]
\textbf{Step 2.} Preprocessing: Remove link $k_0$ iteratively
according to \eqref{smart} until condition \eqref{necessary} holds
true.\\[2.5pt]
\textbf{Step 3.} Power control: Compute parameter $\alpha$ by
\eqref{optalpha} and solve problem \eqref{relax}; check whether all
links are supported: if yes, go to \textbf{Step 5}; else
go to \textbf{Step 4}. \\[2.5pt]
\textbf{Step 4.} Admission control: Remove link  $k_0$ according to
\eqref{removalpower}, set ${\K}={\K}\setminus\left\{k_0\right\},$ and go to
\textbf{Step 3}.\\[2.5pt]
\textbf{Step 5.} Postprocessing: Check the removed links for
possible admission.
\end{minipage}
\end{minipage}
}
\end{center}

In the above framework, if the power control problem \eqref{relax} is solved via {solving its equivalent SOCP reformulation \eqref{relax2}}, we call the corresponding algorithm SOCP-D; while if problem \eqref{relax} is solved via {using the PABB algorithm to solve its smoothing counterpart \eqref{relax-appro}}, we call the corresponding algorithm PABB-D.} The SOCP-D algorithm is of polynomial time complexity, i.e., it has a complexity of ${}{O(N^{3.5}K^{4.5})},$ since it needs to solve at most $K$  SOCP problems \eqref{relax2} and {solving one SOCP problem in the form of \eqref{relax2} requires ${}{O(N^{3.5}K^{3.5})}$ operations} \cite[Page 423]{com3}. {}{It is hard to analyze the complexity of the PABB-D algorithm. This is because global (linear) convergence rate of the PABB algorithm, when it is used to solve general nonlinear optimization problems, remains unknown \cite{bb}.} The postprocessing step (\text{Step 5}) aims at admitting the
links removed in the preprocessing and admission control steps\cite{convex_approximation,msp}. A specification of the postprocessing step can be found in Section III of \cite{technicalreport}.


%
%
%
%
%
%
%

\section{{}{Numerical Simulations}}\label{sec-simulation}
To illustrate the effectiveness and efficiency of the two proposed convex approximation-based deflation algorithms (SOCP-D and PABB-D), we present some numerical simulation results in this section. The number of supported links, the total
transmission power, and the execution CPU time are used as the
metrics for {comparing} different algorithms.

{\emph{Simulation Setup:} 
As in \cite{convex_approximation}, {}{each transmitter's location
obeys the uniform distribution over a $D_1$ Km~$\times$~$D_1$ Km square
and the location of each receiver is uniformly generated in a disc
with center at its corresponding transmitter and radius $D_2$ Km,
excluding a radius of $10$ m.} Suppose that the channel coefficient $h_{k,j}$ is generated from the Rician channel model\cite{goldsmith}, i.e.,
\begin{equation}\label{channel}h_{k,j}=\left(\sqrt{\frac{\kappa}{\kappa+1}}+\sqrt{\frac{1}{\kappa+1}}\zeta_{k,j}\right)\frac{1}{d_{k,j}^2},~\forall~k,~j\in\K,\end{equation}
where $\zeta_{k,j}$ obeys the standard complex Gaussian distribution, i.e., $\zeta_{k,j}\sim{\cal{CN}}({0}, 1),$ $d_{k,j}$ is the Euclidean distance from the link of transmitter $j$ to the link of
receiver $k,$ and $\kappa$ is the ratio of the power in the line of sight (LOS) component to the power in the other (non-LOS)
multipath components. For $\kappa=0$ we have Rayleigh fading and for $\kappa=\infty$ we have no
fading (i.e., a channel with no multipath and only a LOS component). The parameter $\kappa$ therefore is a measure of the severity of the channel fading: a small $\kappa$ implies severe fading and a large $\kappa$ implies relatively mild fading. The channel gain $\left\{g_{k,j}\right\}$~are set to be:
\begin{equation}\label{gain}{}{g_{k,j}=\left|h_{k,j}\right|^2=\left|\sqrt{\frac{\kappa}{\kappa+1}}+\sqrt{\frac{1}{\kappa+1}}\zeta_{k,j}\right|^2\frac{1}{d_{k,j}^4},~k,~j\in\K.}\end{equation}
%
Each link's SINR target is set to be
$\gamma_k=2~\text{dB}~(\forall~k\in\K),$ each link's noise power is set to be $\eta_k=-90~\text{dB}~(\forall~k\in\K),$ and the power budget of
the link of transmitter $k$ is set to be \begin{equation}\label{powerbudget}\bar p_k=b\underline p_k,~k\in\K,\end{equation} where $\underline p_k$ is the minimum power needed
by link $k$ to meet its SINR requirement in the absence of any interference
from other links when $\kappa=+\infty$ in \eqref{channel}.} 

{\emph{Benchmark:} When $\kappa=+\infty$, there is no uncertainty of channel gains, and the number of supported links in this case should be greater than or equal to the number of supported links under the same channel conditions except where $\kappa<+\infty$. In addition, if the number of supported links under these two cases are equal to each other, the total transmission power in the former channel condition should be less than the one in the latter channel condition. In fact, when $\kappa=+\infty$, the corresponding JPAC problem \eqref{csparse2} reduces to problem \eqref{sparse2}, which can be solved efficiently by the NLPD algorithm in \cite{msp}. The solution given by the NLPD algorithm will be used as the benchmark to compare with the two proposed algorithms\footnote{{We remark that this is the first paper that addresses the JPAC problem based on the CDI assumption without specifying any particular distribution, and there is no existing algorithms dealing with the same issue that we can compare the proposed algorithms with.}}, since the NLPD algorithm {was} reported to have the close-to-global-optimal performance {}{in terms of the number of supported links} in \cite{msp}.}



\emph{Choice of Parameters:} We set the parameters $\epsilon,~\delta,$ and $K$ in \eqref{N-K*} to be $0.1,$ $0.05,$ and $10,$ respectively. {We remark that $K$ in equation \eqref{N-K*} is the number of supported links but not the number of total links.} Substituting these parameters in \eqref{N-K*}, we obtain $N^*=200,$ and we set $N=200$ in all of our simulations. Both of the parameters $c_1$ and $c_2$ in \eqref{optalpha} are set to be $0.999.$
  {}{We do simulations in two different setups where $(D_1,D_2)=(2,0.4)$ and $(D_1,D_2)=(1,0.2).$ For convenience, we call the former setup as Setup1 and the latter one as Setup2. Notice that Setup2 represents a dense network where the distance between the transmitters and receivers are closer (i.e., half of that of the Setup1). Under each setup, {we test three different sets of parameters, where one is $(\kappa,b)=(+\infty,2),$ one is $(\kappa,b)=(100,4),$ and another one is $(\kappa,b)=(10,40).$}}~
Finally, we use CVX\cite{cvx} to solve the SOCP problems in the SOCP-D algorithm.



{\begin{table*}[!ht]
\caption{{}{Statistics of the Number of Supported Links of $200$ Monte-Carlo Runs.}}\label{table1}\centering
 \begin{tabular}{|c|c|c|c|}

  \hline\hline
   Parameters $(K,D_1,D_2,\kappa,b)$ & Algorithm  & Statistics of the Number of Supported Links \\ \hline
     $(4,2,0.4,+\infty,2)$ &Benchmark & 664=2*19+3*98+4*83 \\  \hline
     $(4,2,0.4,100,4)$ &SOCP-D/PABB-D & 659=2*19+3*103+4*78 \\ \hline
     $(4,2,0.4,10,40)$ &SOCP-D/PABB-D & 609=1*1+2*39+3*110+4*50 \\ \hline

     $(12,2,0.4,+\infty,2)$ &Benchmark & 1468=5*8+6*36+7*68+8*60+9*24+10*4 \\ \hline
     $(12,2,0.4,100,4)$ &SOCP-D/PABB-D & 1431=5*12+6*42+7*72+8*54+9*17+10*3  \\ \hline
     $(12,2,0.4,10,40)$ &SOCP-D/PABB-D & 1236=4*10+5*39+6*79+7*53+8*15+9*4  \\ \hline

     $(20,2,0.4,+\infty,2)$ &Benchmark & 1953=6*1+7*8+8*21+9*50+10*67+11*38+12*11+13*3+14*1 \\ \hline
     $(20,2,0.4,100,4)$ &SOCP-D/PABB-D & 1864=7*13+8*40+9*58+10*58+11*22+12*8+13*1 \\ \hline
    $(20,2,0.4,10,40)$ &SOCP-D/PABB-D & 1570=5*2+6*17+7*65+8*60+9*39+10*15+11*2 \\ \hline

    $(28,2,0.4,+\infty,2)$ &Benchmark  & 2342=8*1+9*13+10*23+11*54+12*56+13*31+14*17+15*2+16*1+17*2 \\ \hline
   $(28,2,0.4,100,4)$ &SOCP-D/PABB-D & 2250=8*3+9*21+10*40+11*50+12*50+13*22+14*11+15*1+16*2\\ \hline
     $(28,2,0.4,10,40)$ &SOCP-D/PABB-D & 1875=5*1+6*1+7*13+8*37+9*60+10*48+11*25+12*14+13*0+14*1\\ \hline

    $(4,1,0.2,+\infty,2)$ &Benchmark & 639=2*26+3*109+4*65 \\  \hline
     $(4,1,0.2,100,4)$ &SOCP-D/PABB-D & 632=1*1+2*27+3*111+4*61 \\ \hline
    $(4,1,0.2,10,40)$ &SOCP-D/PABB-D & 589=1*2+2*44+3*117+4*37 \\ \hline

     $(12,1,0.2,+\infty,2)$ &Benchmark & 1443=4*3+5*4+6*38+7*85+8*45+9*22+10*3 \\ \hline
     $(12,1,0.2,100,4)$ &SOCP-D/PABB-D & 1403=4*3+5*11+6*45+7*82+8*41+9*16+10*2  \\ \hline
    $(12,1,0.2,10,40)$ &SOCP-D/PABB-D & 1214=3*2+4*8+5*56+6*64+7*52+8*14+9*4  \\ \hline

     $(20,1,0.2,+\infty,2)$ &Benchmark & 1942=6*1+7*3+8*29+9*49+10*69+11*37+12*11+13*1 \\ \hline
    $(20,1,0.2,100,4)$ &SOCP-D/PABB-D & 1882=6*1+7*10+8*32+9*65+10*58+11*23+12*11 \\ \hline
    $(20,1,0.2,10,40)$ &SOCP-D/PABB-D & 1577=5*4+6*18+7*54+8*59+9*55+10*6+11*4 \\ \hline

     $(28,1,0.2,+\infty,2)$ &Benchmark  & 2333=8*2+9*13+10*34+11*41+12*50+13*37+14*18+15*4+16*1 \\ \hline
     $(28,1,0.2,100,4)$ &SOCP-D/PABB-D & 2236=7*1+8*5+9*22+10*36+11*57+12*38+13*27+14*13+15*1\\ \hline
     $(28,1,0.2,10,40)$ &SOCP-D/PABB-D & 1857=6*5+7*16+8*38+9*51+10*51+11*28+12*9+13*2\\ \hline
  \hline
\end{tabular}\end{table*}}

\emph{Simulation Results and Analysis:} {Table \ref{table1} summarizes the statistics of the number of supported links of $200$ Monte-Carlo runs of numerical experiments with different choices of simulation parameters. {}{For instance, ``$664=2*19+3*98+4*83$'' in the third column of Table \ref{table1} stands for that when $\left(K,D_1,D_2,\kappa,b\right)=\left(4,2,0.4,+\infty,2\right),$ total $664$ links are supported in these $200$ Monte-Carlo runs, and amongest them, $2$ links are supported $19$ times, $3$ links are supported $98$ times, and $4$ links are supported $83$ times.} Figs. \ref{user1}, \ref{power1}, and \ref{time1} are obtained by averaging over the $200$ Monte-Carlo runs. 
They plot the average number of supported links, the average total transmission power, 
 and the average execution CPU time of the proposed SOCP-D and PABB-D algorithms {}{(for solving the sampled JPAC problem \eqref{csparse2})} and the benchmark versus different number of total links in Setup1.


\begin{figure}[!t]
     \centering
     \includegraphics[width=8.8cm]{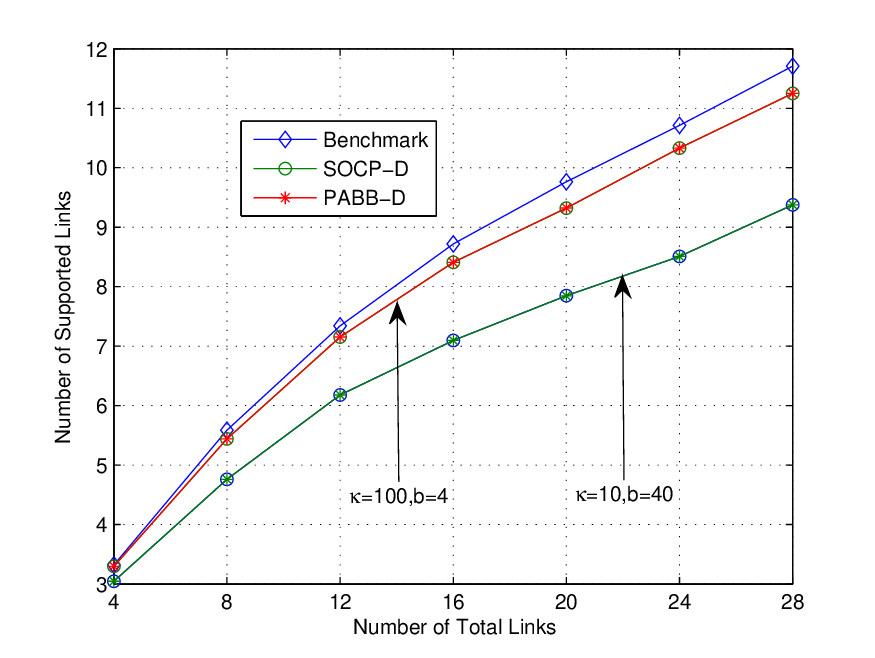}
     \caption{{}{Average number of supported links versus the number of total links in Setup1.}}
     \label{user1}
     \end{figure}

 \begin{figure}[!t]
     \centering
     \includegraphics[width=8.8cm]{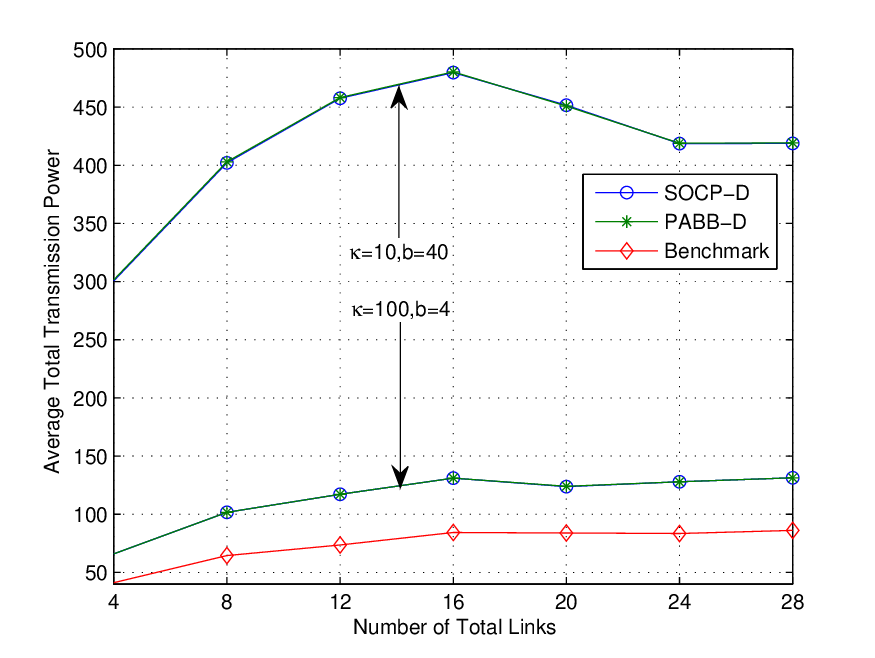}
     \caption{{}{Average total transmission power versus the number of total links in Setup1.}}
     \label{power1}
     \end{figure}
     \begin{figure}[!t]
     \centering
     \includegraphics[width=8.8cm]{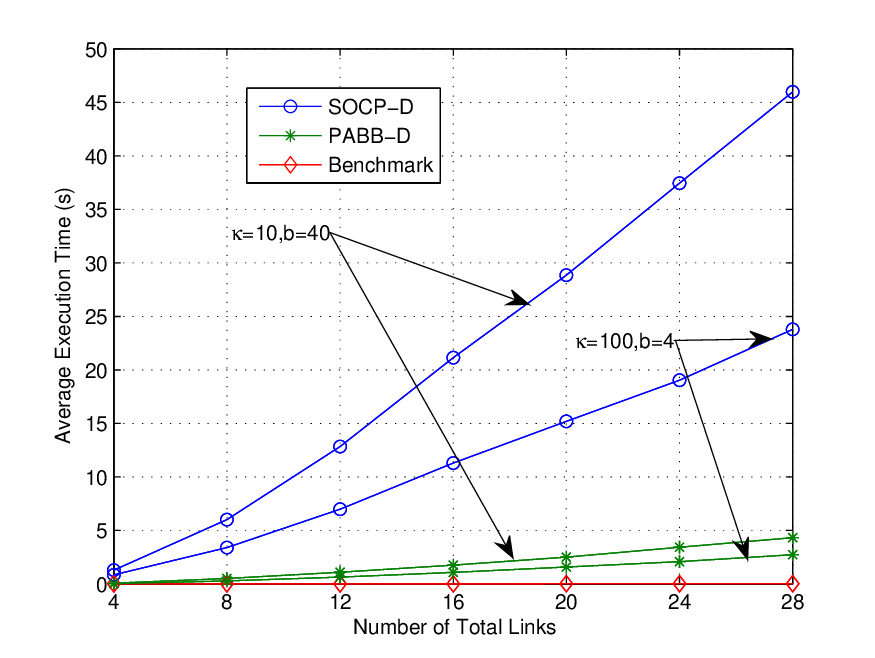}
     \caption{{}{Average execution time (in seconds) versus the number of total links in Setup1.}}
     \label{time1}
     \end{figure}

{It can be seen from Fig. \ref{user1} that the number of supported links by the two proposed algorithms (for fading channels) is less than the benchmark (for deterministic channels). This shows that the uncertainty of channel gains could lead to a (significant) reduction in the number of supported links. This can also be clearly observed from Table \ref{table1}. For instance, when $K=4$ (see the first three lines of Table \ref{table1}) all links can be simultaneously supported $83$ times when $\kappa=+\infty,$ $78$ times when $\kappa=100,$ and only $50$ times when $\kappa=10.$ In fact, this is the reason why we associate different $\kappa$ with different $b$ in our simulations. We expect that a large $b$ and thus large power budgets $\bar p_k$ (cf. \eqref{powerbudget}) can compensate the performance degradation of the number of supported links caused by the large uncertainty of channel gains.}

%
%
{Table \ref{table1}, Fig. \ref{user1}, and Fig. \ref{power1} show that the two proposed algorithms always return the same solution to the sampled JPAC problem \eqref{csparse2}, i.e., supporting same number of links with same total transmission power. However, Fig. \ref{time1} shows that the PABB-D algorithm substantially outperforms the SOCP-D algorithm in terms of the average CPU time. This is not surprising, {}{since in each power control step (i.e., solving the convex approximation problem \eqref{relax}), the custom-design algorithm is used to carry out power control in the PABB-D algorithm while a general-purpose solver CVX is used to update power in the SOCP-D algorithm. Note that both the number of constraints and the number of unknown variables in problem \eqref{relax2} are of order $O(NK)$ and \eqref{N-K*} shows that the sample size $N$ needs to be large to guarantee the approximation performance, which makes CVX unsuitable to be used to solve problem \eqref{relax} via solving {}{its} equivalent SOCP {}{reformulation} \eqref{relax2}.}}

\begin{figure}[!t]
     \centering
     \includegraphics[width=8.8cm]{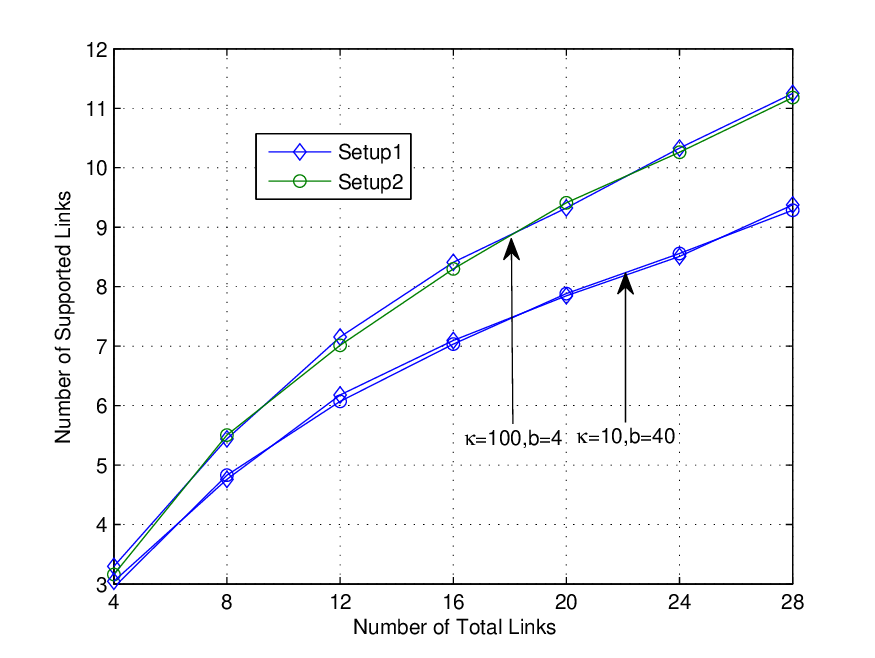}
     \caption{{}{Average number of supported links versus the number of total links in Setup1 and Setup2.}}
     \label{nuser1}
     \end{figure}

      \begin{figure}[!t]
     \centering
     \includegraphics[width=8.8cm]{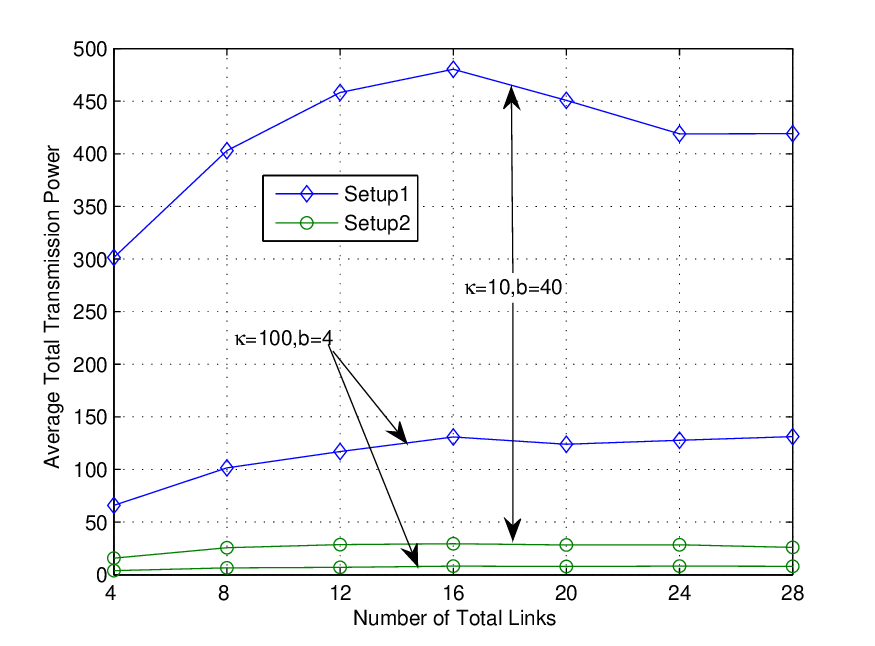}
     \caption{{}{Average total transmission power versus the number of total links in Setup1 and Setup2.}}
     \label{npower1}
     \end{figure}

{By comparing the two sets of numerical experiments where $(\kappa,b)=(100,4)$ and $(\kappa,b)=(10,40),$ it can be observed from Figs. \ref{user1} and \ref{power1} that more links can be supported with significantly less total transmission power in the former case than the latter case.%
%
~This is because the uncertainty of the channel gains with $\kappa=100$ is generally much smaller than the one with $\kappa=10.$ 
We also point out that the execution CPU time of the two proposed deflation algorithms mainly depends on how many times the power control problem \eqref{relax} is solved. In general, the larger number of links are supported, the smaller number of links are removed from the network and the smaller number of power control problems in form of \eqref{relax} are solved. Therefore, the average CPU time of the proposed algorithms when $(\kappa,b)=(10,40)$ is larger than the one when $(\kappa,b)=(100,4);$ see Fig. \ref{time1}.}

{}{For conciseness, we do not present the comparison of the PABB-D and SOCP-D algorithms when both algorithms are used to solve the sampled JPAC problem \eqref{csparse2} in Setup2, since the same observations as in Setup1 can be made in Setup2. Instead, we focus on the comparison of two different setups in the following. Since the PABB-D algorithm always returns the same solution as the SOCP-D algorithm but takes much less CPU time, we choose to use the PABB-D algorithm to solve the sampled JPAC problem \eqref{csparse2} in the following.}

\begin{table*}[!ht]
\caption{{}{The Ratio of the Average Total Transmission Power in Setup1 to That in Setup2.}}\label{table3}\centering
 \begin{tabular}{|c|c|c|c|c|c|c|c|}
  \hline\hline
   \backslashbox{$(\kappa,b)$}{$K$} & 4 & 8 & 12 & 16 & 20 & 24 & 28 \\\hline
  $(100,4)$ & 17.1513 & 15.7554 & 16.5563 & 16.1834 & 15.9193 & 15.6577 & 16.4100 \\\hline
  $(10,40)$ & 19.1574 & 15.7159 & 16.0476 & 16.2672 & 15.9009 & 14.7517 & 16.0995 \\
  \hline\hline
\end{tabular}\end{table*}

{Figs. \ref{nuser1} and \ref{npower1} report the average number of supported links and the average total transmission power ({}{returned by the PABB-D algorithm for solving the sampled JPAC problem \eqref{csparse2}}) versus different number of total links in Setup1 and Setup2. It can be observed that the average number of supported links and the average execution time in both setups are roughly equal to each other, but the average total transmission power in Setup1 is approximately $16$ times as large as that in Setup2 (also see Table \ref{table3}). This is because when the setup is switched from Setup1 to Setup2 (with the random variables $\zeta_{k,j}$ in \eqref{channel} and \eqref{gain} neglected for the time being), the corresponding distances $d_{k,j}$ between the transmitters and receivers decrease by half. {}{According to \eqref{gain} and \eqref{powerbudget}, all the channel gains (including both the direct-link and cross-link channel gains) increase and power budgets decrease by a factor of $16.$ As the channel gains are increased and power budgets are decreased by a same factor while the noise powers remain to be fixed, the number of supported links in problem \eqref{csparse2} remains unchanged.}
However, it brings a benefit of a 93.75\% (=15/16) reduction in the total transmission power, which is consistent with our engineering practice. Due to the effects of the random variables $\zeta_{k,j}$ in \eqref{channel} and \eqref{gain}, the ratio of the average total transmission power in Setup1 to that in Setup2 is approximately (but not exactly) $16.$
}

\section{Conclusions}
In this paper, we have considered the chance SINR constrained JPAC problem, 
and have proposed two sample approximation-based deflation approaches for {solving} the problem. We first approximated the computationally intractable chance SINR constraints by sampling, and then reformulated the sampled JPAC problem as a composite group sparse minimization problem. Furthermore, we approximated the NP-hard group sparse minimization problem by a convex problem (equivalent to an SOCP) and used its solution to check the simultaneous supportability of all links and to guide an iterative link removal procedure (the deflation {}{approach}), resulting in two efficient deflation algorithms (SOCP-D and PABB-D).

The proposed approaches are particularly attractive for practical implementations for the following reasons. {First, the two proposed approaches only require the CDI, which is more practical than most of the existing algorithms for JPAC where the perfect instantaneous CSI is required.} Second, the two proposed approaches enjoy a low computational complexity. The SOCP-D approach has a polynomial time complexity. To further improve the computational efficiency, the special structure of the SOCP approximation problem is exploited, and an efficient algorithm, {PABB-D, is custom designed for solving it. The PABB-D algorithm} significantly outperforms the SOCP-D algorithm in terms of the execution CPU time. Finally, our simulation results show that the proposed approaches are very effective by using the NLPD algorithm as the benchmark.

%
%



\appendices

\section{}\label{app-proposition}
 \emph{Proof of Proposition \ref{mingti1}:} 
 We prove Proposition \ref{mingti1} by contraction. Assume that link $k$ is supported in \eqref{csparse2} but $\bA_k\bq^*>\bc_k$ holds true{. Then} we can construct a feasible point $\hat \bq$ 
 satisfying
 \begin{equation}\label{firstsecond}
 \begin{array}{rl}
 &\!\!\sum_{k\in\K}\|\max\left\{\bc_k-\bA_k\hat\bq,\bm{0}\right\}\|_{0}+\alpha\bar\bp^T\hat\bq<\sum_{k\in\K}\|\max\left\{\bc_k-\bA_k\bq^*,\bm{0}\right\}\|_{0}+\alpha\bar\bp^T\bq^*.
 \end{array}\end{equation}
Define $\hat \bq=\left(\hat q_1,\hat q_2,\ldots,\hat q_K\right)^T$ with 
\begin{equation*}
\hat q_{j}=\left\{\begin{array}{ll}
\max\left\{(\bE_k-\bA_k)\bq^*+\bc_k\right\},&\text{if~}j=k;\\[2pt]
q_j^*,&\text{if~}j\neq
k.
\end{array}
\right.
\end{equation*}
Recalling the {definitions} of $\bE_k$ and $\bA_k,$
~we know $\bE_k-\bA_k$ is a nonnegative matrix, and thus $\hat q_k=\max\left\{(\bE_k-\bA_k)\bq^*+\bc_k\right\}>0.$
Since $$\bA_k\bq^*=\bE_k\bq^*-(\bE_k-\bA_k)\bq^*= q_k^* \be-(\bE_k-\bA_k)\bq^*>\bc_k,$$ it follows from the definition of $\hat q_k$ that $$q_k^*>\max\left\{(\bE_k-\bA_k)\bq^*+\bc_k\right\}=\hat q_k.$$ Hence, $\hat \bq$ is feasible (i.e., $\bm{0}\leq \hat \bq\leq \bq^* \leq \be$ and the inequality $\hat \bq\leq \bq^*$ holds true strictly for the $k$-th entry) and \begin{equation}\label{second}\bar\bp^T\hat \bq<\bar\bp^T\bq^*.\end{equation}

Moreover, 
it follows from the definition of $\hat\bq$ that $\bA_k\hat\bq\geq\bc_k.$ For any $j\neq k,$ if $\bA_j\bq^*\geq\bc_j,$ we have $\bA_j\hat\bq-\bc_j\geq \bA_j\bq^*-\bc_j\geq\bm{0},$ where the first inequality is due to the fact that all entries of $\bA_j$ except its $j$-th column are nonpositive and the fact $\hat \bq\leq \bq^*.$
%
Consequently, there holds $\J^*\subset\hat\J,$ where $\J^*=\left\{j\,|\,\bA_j\bq^*\geq\bc_j\right\}$ and $\hat \J=\left\{j\,|\,\bA_j\hat\bq\geq\bc_j\right\}.$ Thus, we have
          \begin{equation}\label{first}\begin{array}{rl}
          \sum_{k\in\K}\|\max\left\{\bc_k-\bA_k\bq^*,\bm{0}\right\}\|_{0}=& \sum_{k\notin \J^*}\|\max\left\{\bc_k-\bA_k\bq^*,\bm{0}\right\}\|_{0}\\[5pt]
          \geq & \sum_{k\notin \hat\J}\|\max\left\{\bc_k-\bA_k\hat\bq,\bm{0}\right\}\|_{0}\\[5pt]
          =&\sum_{k\in\K}\|\max\left\{\bc_k-\bA_k\hat\bq,\bm{0}\right\}\|_{0}.
          \end{array}
          \end{equation}
Combining \eqref{second} and \eqref{first} yields \eqref{firstsecond},
which contradicts the optimality of $\bq^*.$ This completes the proof of Proposition \ref{mingti1}.

\section{}\label{app-1}
\emph{Proof of Proposition \ref{subgradient2}:} To prove Proposition \ref{subgradient2}, we first consider the simple case where $h(\bq)=\|\max\left\{\bq,
  \bm{0}\right\}\|_2.$
\begin{yinli}\label{subgradient}
  Suppose $h(\bq)=\|\max\left\{\bq,
  \bm{0}\right\}\|_2$. Then $\partial h(\bm{0})=\left\{\bs\,|\,\bs\geq\bm{0}, \|\bs\|_2\leq 1\right\}.$ If 
  there exists $i$ such that $(\bq)_i>0,$
  then $h(\bq)$ is differentiable and
  \begin{equation}\label{grad}\nabla h(\bm{\bq})=\frac{\max\left\{\bq,\bm{0}\right\}}{\|\max\left\{\bq,
  \bm{0}\right\}\|_2}.\end{equation}
\end{yinli}
\begin{proof}
  Define ${\S}=\left\{\bs\,|\,\bs\geq\bm{0}, \|\bs\|_2\leq 1\right\}.$ We claim $\partial h(\bm{0})={\S}.$ On one hand, {taking} any $\bs\in{\S},$ we have that $$h(\bq)=\|\max\left\{\bq,
  \bm{0}\right\}\|_2\geq\bs^T\max\left\{\bq,\bm{0}\right\}\geq \bs^T\bq=h(\bm{0})+\bs^T(\bq-\bm{0}),~\forall~\bq,$$ where the first inequality is due to the Cauchy-Schwarz inequality and the fact $\|\bs\|_2\leq 1$, and the second inequality is due to $\bs\geq\bm{0}$. This shows that ${\S}\subset \partial h(\bm{0})$ according to the definition of $\partial h(\bm{0})$ \cite{analysis}. On the other hand, to show $\partial h(\bm{0})\subset {\S},$ it {suffices} to show that any point $\bs\notin{\S}$ is not a subgradient of $h(\bq)$ at point $\bm{0}.$ In particular, if $\|\bs\|_2>1,$ then $$h(\bq)=\|\max\left\{\bq,
  \bm{0}\right\}\|_2\leq\|\bq\|_2=1< \|\bs\|_2= \bs^T\bq=h(\bm{0})+\bs^T(\bq-\bm{0})$$ with $\bq=\bs/\|\bs\|_2.$ Thus, the subgradient $\bs$ of $h(\bq)$ at point $\bm{0}$ must satisfy $\|\bs\|_2\leq1.$ If $\|\bs\|_2\leq 1$ but $(\bs)_1<0$ (without loss of generality), we test the point $\bq=\left(-1,0,\ldots,0\right)^T,$ and obtain
  $$h(\bq)=\|\max\left\{\bq,
  \bm{0}\right\}\|_2=0<-\left(\bs\right)_1= \bs^T\bq=h(\bm{0})+\bs^T(\bq-\bm{0}).$$ Consequently, the subgradient $\bs$ of $h(\bq)$ at point $\bm{0}$ must satisfy $\|\bs\|_2\leq1$ and $\bs\geq\bm{0}.$ Hence, $\partial h(\bm{0})=\left\{\bs\,|\,\bs\geq\bm{0}, \|\bs\|_2\leq 1\right\}.$

  Next, we show that $h(\bq)$ is differentiable at the point $\bq$ {which has at least one positive entry}, and the corresponding gradient is given in \eqref{grad}. In fact, although the function $\max\left\{q,0\right\}$ is nondifferentiable at point $q=0,$ its square $f(q)=\left(\max\left\{q,0\right\}\right)^2$ is differentiable everywhere; i.e.,
  \begin{equation*}
f'({q})=\left\{\begin{array}{cl}
0,&\text{if~}q \leq{0};\\[5pt]
\displaystyle 2q,&\text{if~} q> {0}.
\end{array}
\right.
\end{equation*}
According to the composite rule of differentiation, we know that the gradient of $h(\bq)$ is given by \eqref{grad}. This completes the proof of Lemma \ref{subgradient}.
\end{proof}

Equipping with Lemma \ref{subgradient}, we now can prove Proposition \ref{subgradient2}. Without loss of generality, assume that $\bc_k-\bA_k\bar\bq=\bm{0}$. Then, for any $\bq$ and any $\bs$ satisfying $\|\bs\|_2\leq 1$ and $\bs\geq\bm{0},$ we have
\begin{equation}\label{subdifferential}
\begin{array}{rcl}
h_k(\bq)=\|\max\left\{\bc_k-\bA_k\bq,\bm{0}\right\}\|_2&\ge &\bs^T\left(\bc_k-\bA_k\bq\right)~\left(\text{from Lemma \ref{subgradient}}\right)\\[5pt]
&=&h_k(\bar\bq)+\left(-\bA_k^T\bs\right)^T\left(\bq-\bar\bq\right),
\end{array}
\end{equation}which shows that all vectors in $\left\{-\bA_k^T\bs\,|\,\bs\geq\bm{0}, \|\bs\|_2\leq 1\right\}$ are subgradients of $h_k(\bq)$ at point $\bar\bq.$ In the same way as in the proof of Lemma \ref{subgradient}, we can show that if $\bs$ does not satisfy $\|\bs\|_2\leq 1$ and $\bs\geq\bm{0},$ the inequality in \eqref{subdifferential} will violate for some special {}{choice} of $\bq$. Hence, $\partial h_k(\bm{\bar\bq})=\left\{-\bA_k^T\bs\,|\,\bs\geq\bm{0}, \|\bs\|_2\leq 1\right\}.$

  If $\N_k^+\neq\emptyset,$ we know from the composite rule of differentiation and Lemma \ref{subgradient} that
  $h_k(\bm{\bar\bq})$ is differentiable and its gradient is given by
  \begin{align*}\nabla h_k(\bm{\bar\bq})&=\frac{-\sum_{n\in\N_k^+}\left(\bc_k-\bA_k\bar\bq\right)_n\left(\bm{a}_k^n\right)^T}{\|\max\left\{\bc_k-\bA_k\bar\bq, \bm{0}\right\}\|_2}=\frac{-\bA_k^T\max\left\{\bc_k-\bA_k\bar\bq, \bm{0}\right\}}{\|\max\left\{\bc_k-\bA_k\bar\bq, \bm{0}\right\}\|_2}.\end{align*} This completes the proof of Proposition \ref{subgradient2}.

\section{}\label{app-neverover}
  \emph{Proof of Theorem \ref{neverover}:} Suppose all links in the network can be simultaneously supported (i.e., there exists $\bm{0}\leq \bq\leq \be$ satisfying $\bA\bq\geq\bc$) and $\bar\bq$ is the solution to problem
  \begin{equation*}
\begin{array}{cl}
\displaystyle \min_{\bq} & \displaystyle \bar\bp^T\bq \\
  \mbox{s.t.} &  \bA\bq-\bc\geq \bm{0},\\
&\displaystyle \bm{0}\leq \bq\leq \be.
\end{array}
\end{equation*}To prove Theorem \ref{neverover}, it {suffices} to show that $\bar\bq$ is also the solution to problem \eqref{relax} with $\alpha\in(0,\alpha_2]$.
Moreover, to show $\bar \bq$ is the solution to problem \eqref{relax}, we only need to show that the subdifferential of the objective function of problem \eqref{relax} at point $\bar\bq$ contains $\bm{0}$\cite{analysis}. Next, we claim the latter is true.

%

We first characterize the subdifferential of the objective function of problem \eqref{relax} at point $\bar\bq$. It follows from \eqref{balance} that there {exists} ${\I}=\left\{n_1,n_2,\ldots,n_K\right\}$ such that $\bar\bq$ is the solution to the following linear system $$\bA_{\I}\bq:=\left[\bm{a}_1^{n_1};\bm{a}_2^{n_2};\ldots;\bm{a}_K^{n_K}\right]\bq=\left(c_1^{n_1};c_2^{n_2};\ldots;c_K^{n_K}\right):=\bc_{\I}.$$
 Recalling the definition of $\bm{a}_k^n$ (see Subsection \ref{sec-normalization}), we know that $\bI-\bA_{\I}$ is a nonnegative matrix. Moreover, from \cite[Theorem 1.15]{matrix}, $\bA_{\I}$ is nonsingular, $\bA_{\I}^{-1}$ is nonnegative, and \begin{equation}\label{upper}0<\be^T\bA_{\I}^{-1}\bc_{\I}=\be^T\bar\bq\leq K.\end{equation}
Define $h_k(\bq)=\left\|\max\left\{\bc_k-\bA_k\bq,\bm{0}\right\}\right\|_{2}$ for $k\in\K.$ It follows from \cite[Theorem 23.8]{analysis} that the subdifferential of the objective function of problem \eqref{relax} at point $\bar\bq$ is given by
$$\left\{\sum_{k\in\K}\bg_k+\alpha\bar\bp\,|\,\bg_k\in\partial h_k(\bar\bq),\,k\in\K\right\}.$$
%
According to Proposition \ref{subgradient2}, $\partial h_k(\bar\bq)$ \emph{contains}\footnote{The subdifferential of $h_k(\bq)$ at point $\bar\bq$ is not necessarily equal to ${\S}_k$ in \eqref{Sk}. This is because that some other entries (except the $n_k$-th entry) of the vector $\bc_k-\bA_k\bar\bq$ might also be zero.} all vectors in
\begin{equation}\label{Sk}{\S}_k=\left\{-s_k\left(\bm{a}_k^{n_k}\right)^T\,|\,0\leq s_k\leq 1\right\}.\end{equation}
   Therefore, all vectors in \begin{align*}{\S}=\left\{-\bA_{\I}^T\bs+\alpha\bar\bp\,|\,\bm{0}\leq\bs\leq\be\right\}=\left\{-\sum_{k\in\K}s_k\left(\bm{a}_k^{n_k}\right)^T+\alpha\bar\bp\,|\,0\leq s_k\leq1, k\in\K\right\}
   \end{align*} are subgradients of the objective function of problem \eqref{relax} at point $\bar\bq.$

   If $\bm{0}\in\S,$ the subdifferential of the objective function of problem \eqref{relax} at point $\bar\bq$ contains $\bm{0}$\cite{analysis}, which completes the proof of Theorem \ref{neverover}. Next, we show $\bm{0}\in\S$ is true. 
   Consider the vector $\bs=\alpha\bA_{\I}^{-T}\bar\bp.$ It is a nonnegative vector (since $\bA_{\I}^{-1}$ is nonnegative), and each of its entries is less than or equal to $1$ as long as $\alpha\leq\alpha_2.$ This is because
    \begin{equation*}
      \begin{array}{rcl}
        \be^T\bs=\be^T\alpha\bA_{\I}^{-T}\bar\bp\leq\alpha\max\left\{\bar\bp\right\}\be^T\bA_{\I}^{-T}\be                    \leq\alpha\frac{\max\left\{\bar\bp\right\}}{\min\left\{\bc_{\I}\right\}}\bc_{\I}^T\bA_{\I}^{-T}\be                                        \overset{(a)}{\leq}\alpha\frac{\max\left\{\bar\bp\right\}}{\min\left\{\bc\right\}}K&\leq&1,
      \end{array}
    \end{equation*}where $(a)$ is due to \eqref{upper} and the fact $\min\left\{\bc\right\}\leq\min\left\{\bc_{\I}\right\}$.
Substituting $\bs=\alpha\bA_{\I}^{-T}\bar\bp$ into ${\S},$ we obtain 
$-\bA_{\I}^T\bs+\alpha\bar\bp=-\bA_{\I}^T\left(\alpha\bA_{\I}^{-T}\bar\bp\right)+\alpha\bar\bp=\bm{0}.$ Thus, $\bm{0}\in\S.$

 \vspace*{4pt}

{}{\section*{Acknowledgment}
The authors would like to thank Professor Yu-Hong Dai of the Chinese Academy of Sciences for many
useful discussions and the anonymous reviewers for their useful comments.}

\end{document}